\documentclass[11pt,twoside]{article}

\usepackage[ansinew]{inputenc} 
\usepackage{amsmath,amsfonts,amssymb,amsthm}
\usepackage{pstricks,pst-node,pst-coil,pst-plot,pstricks-add}
\usepackage{geometry,epsfig}
\usepackage{bbm}
\usepackage{cite}
\usepackage{graphicx}

\usepackage{bbm}
\usepackage{cite}
\usepackage{paralist}
\usepackage{cleveref}
\usepackage{enumitem}
\usepackage{emptypage}

\DeclareGraphicsExtensions{.pdf}

\usepackage{mathrsfs} 

\usepackage[T1]{fontenc}

\newcommand{\field}[1]{\mathbb{#1}}
\newcommand{\N}{\field{N}}
\newcommand{\R}{\field{R}}

\newcommand{\HH}{\mathscr H}

\newcommand{\LL}{\mathscr L}
\newcommand{\FF}{\mathcal F}

\newcommand{\ph}{\varphi}
\newcommand{\const}{\mathrm{const}}

\renewcommand{\ker}{\mathrm{Ker}}

\newcommand{\norm}[1]{\mbox{$\left\| #1 \right\|$}}           
\newcommand{\sprod}[2]{\langle #1,#2 \rangle}        
\newcommand{\shifted}[1]{\widehat{#1}}
\newcommand{\shiftedB}[1]{\widetilde{#1}}

\newcommand{\sgn}{\operatorname{sgn}}

\newtheorem{thm}{Theorem}

\newtheorem{cor}[thm]{Corollary}
\newtheorem{lemma}[thm]{Lemma}

\newtheorem*{remarks}{Remarks}

\font\notefont=cmsl8 
\title{Stability of the two-dimensional Fermi polaron}

\author{Marcel~Griesemer\footnote{marcel.griesemer@mathematik.uni-stuttgart.de}\ and Ulrich~Linden\footnote{ulrich.linden@mathematik.uni-stuttgart.de}\\  
\small Fachbereich Mathematik, Universit\"at Stuttgart, D-70569 Stuttgart, Germany}  

\begin{document}
\maketitle
\begin{abstract}
   A system composed of an ideal gas of $N$ fermions interacting with an impurity particle in two space dimensions is considered. The interaction between impurity and fermions is given in terms of two-body point interactions whose strength is determined by the two-body binding energy, which is a free parameter of the model. If the mass of the impurity is 1.225 times larger than the mass of a fermion, it is shown that the energy is bounded below uniformly in the number $N$ of fermions. This result improves previous, $N$-dependent lower bounds and it complements a recent, similar bound for the Fermi polaron in three space dimensions.
\end{abstract}

\section{Introduction}


The system considered in this paper is composed of an ideal gas of $N$ fermions and one additional particle, called impurity, in two space dimensions. The impurity interacts with the fermions by two-body point interactions. Informally, the Hamiltonian of the system may thus be written as 
\begin{equation} \label{Formal_Expression_1}
   - \frac{1}{M} \Delta_y - \sum_{i=1}^N \Delta_{x_i} - g \sum_{i=1}^N \delta(x_i - y),
\end{equation}
where $M>0$ is the mass of the impurity and $g$ plays the role of a coupling constant. The problem of defining a self-adjoint 
Hamiltonian describing 
\eqref{Formal_Expression_1} is discussed and solved in \cite{Figari,DR}. We are interested in its ground state energy, and we show, in this paper, that it is 
bounded below uniformly in $N$, provided that $M > 1.225$. In the physics literature the system described above is called Fermi polaron 
\cite{Bruun_Review}. It is a model for an ultra-cold gas of fermionic atoms interacting with an additional, impurity atom. One is interested in the 
form of the ground state as a function of the coupling strength and in two space dimensions one expects a sharp transition, related to the BEC-BCS crossover 
\cite{Nature,ParishLevinsen}. 

Our approach for defining a self-adjoint Hamiltonian describing \eqref{Formal_Expression_1} follows \cite{Figari,DR} and it is described in \cite{Thesis}. Here we only summarise the ingredients and facts needed in this paper. Some more details are given in the appendix. We work in the language of second quantisation. Let $a_k$ and $a_k^*$ denote the usual annihilation and creation operators in antisymmetric Fock space over $L^2(\R^2)$. Let  $H_f = \int k^2 a_k^* a_k\, dk$ and  $P_f = \int k a_k^* a_k\, dk$. The main operator to be analysed in this paper is not the Hamiltonian of the system, which is not available explicitly, but a self-adjoint operator $\phi(E)$ in $\HH_{N-1}=\bigwedge^{N-1}L^2(\R^2)$, depending on a parameter $E<0$ and defined by
\begin{equation} \label{def:Phi}
   \phi(E) = \alpha + \phi^0(E) + \phi^I(E),
\end{equation}
where $\alpha\in \R$ and 
\begin{align}
  \phi^0(E) &:= \frac{\pi}{1+\frac{1}{M}}\log\left(\frac{1}{M+1} P_f^2 + H_f - E \right)\label{def:Phi_0}\\
  \phi^I(E) &:= \int \! dp \, dq \: a_p^* \: \frac{1}{\frac{1}{M}(P_f + p + q)^2 + H_f + p^2 + q^2 - E} \: a_q.\label{def:Phi_I}
\end{align}
Here $\alpha$ is a free parameter of the model that parametrises the coupling strength between fermions and impurity. To give it a physical interpretation we 
mention that $\alpha = - (\pi/(1+M^{-1})) \cdot \log |E_B|$, where $E_B < 0$ is the ground state energy of the two-body system consisting of only one fermion 
and the impurity. A negative energy state is always present in the two-body system in two dimensions \cite{SolvableModels}. The point 
about $\phi(E)$ is that, for $E < 0$, 
\begin{equation}\label{Criterion_Lower_Bound}
   \phi(E) \geq 0 \quad \Rightarrow \quad H_N \geq E
\end{equation}
where $H_N$ denotes the self-adjoint realisation of \eqref{Formal_Expression_1} after separating off the center-of-mass motion (see Appendix~\ref{sec:appendix}). The main result of this paper, Theorem \ref{Main_Theorem} below, provides us with a number $E$, depending on $M$ only, such that $\phi(E) \geq 0$ for $M>1.225$. This implies, by \eqref{Criterion_Lower_Bound}, that $H_N \geq E$ uniformly in $N$. It is an open problem, whether or not an $N$-independent lower bound on $H_N$ exists for arbitrary positive values of $M$. For a possible approach to this problem see \cite{Thesis}.

For arbitrary $M>0$ the Hamiltonian $H_N$ is bounded below, but the lower bound may depend on $N$: using that $\phi^I(E)$ is bounded with $\norm{\phi^I(E)} \leq 
\const\cdot (N-1)$ and $\phi^0(E) \geq (\pi/(1+M^{-1})) \cdot \log (-E)$, we conclude, by \eqref{Criterion_Lower_Bound}, that
\begin{equation}\label{LB-N}
   H_N \geq E_B \cdot \exp(C \cdot (1 + M^{-1}) \cdot (N-1)).
\end{equation}
This result can already be inferred from \cite{Figari}.

In three space dimensions, \eqref{Criterion_Lower_Bound} still holds with $\phi(E)$ defined by \eqref{def:Phi}, \eqref{def:Phi_I}, see \cite{Figari, MoserSeiringer},  and 
$$
   \phi^0(E) = \frac{2 \pi^2}{(1 + \frac{1}{M})^{3/2}} \cdot \sqrt{\frac{1}{M+1} P_f^2 + H_f - E}.
$$
In contrast to the two-dimensional case, however, the operator $\phi^I(E)$ is not bounded anymore. Instead, from \cite{Five_Italians} it follows that
\begin{equation}\label{phi-bound}
   \phi^{I}(E) \geq -C(M,N) \cdot \phi^0(E),
\end{equation}
which implies $\phi(E) \geq \alpha + (1 - C(M,N)) \phi^0(E)$. Provided that $C(M,N)<1$, we may choose $|E|$ large enough, using that 
$\phi^0(E) \geq (2 \pi^2/(1 + M^{-1})^{3/2})\sqrt{- E}$, such that $\phi(E)\geq 0$ and hence $H_N \geq E$. This improves an earlier result 
\cite{Minlos_N+1,Minlos_N+1+}. The condition $C(M,N)<1$ is satisfied if $M$ is larger than the critical mass $M^*(N)$ defined by $C(M^*(N), N) = 1$. In 
\cite{Five_Italians}, \eqref{phi-bound} is shown for a function $C(M,N)$ for which $M^*(N) \propto N$ as $N\to\infty$, and, moreover, it is shown that 
$\phi(E)$ is unbounded from 
below if $N \geq 2$ and $M < M^*(2) \approx 0.0735$. Recently, \eqref{phi-bound} was shown to hold with a constant $C(M)$ that is independent of $N$ 
\cite{MoserSeiringer}. This constant satisfies $C(M)<1$ if $M>M^*\simeq 0.36$. It follows that $H_N$ is bounded below uniformly in $N$ provided that $M>0.36$. 
In fact, $H_N \geq 0$ if $\alpha \geq 0$ and $H_N \geq - \tilde{C}(M) \cdot \alpha^2$ for a constant $\tilde{C}(M) > 0$ if $\alpha < 0$. While the present 
paper, on a technical level (i.e. in the proof of \Cref{Bound_MS} below), has strongly benefited from \cite{MoserSeiringer}, it solves an additional infrared 
problem, which arises due to the lower dimension. The main result can also be found in \cite{Thesis}.  

\section{An $N$-independent lower bound for the Fermi polaron in $\R^2$}

\begin{thm}\label{Main_Theorem}
 Let $E_B < 0$. Set
\begin{equation}\label{Definition_Alpha}
   \alpha(M) := \frac{1}{2(M+1)} + \frac{1}{2} \int\limits_0^1\frac{1}{\beta(u) (M + 1 - u)}\, du,
\end{equation}
where
\begin{equation}\label{Definition_Beta}
   \beta(u) := \min \left\{ 1, \frac{(M + 1 - u)(M + 2)}{M^2 + 3M + 1 - u} \right\},
\end{equation}
and suppose that $\alpha(M) < M/(M+1)$, which is satisfied if $M > 1.225$. Then, for every $\lambda > 0$, the unique solution $\mu < 0$ of the equation
\begin{equation}\label{Equation_for_Lower_Bound}
  \left( \frac{M}{M+1} - \alpha(M) \right) \log \left( \frac{\mu}{E_B} \right) - \sqrt{\frac{\lambda}{-\mu}} - \sqrt{\frac{\lambda}{\lambda - \mu}} - \alpha(M) \log\left( E_B \left( \frac{1}{\mu} - \frac{1}{\lambda} \right) \right) - \alpha(M) = 0
\end{equation}
satisfies $\phi(\mu)\geq 0$ and hence $H_N \geq \mu$ for all $N \in \N$.
\end{thm}

\begin{remarks} \leavevmode
 \begin{enumerate}[label=(\roman*)]
   \item The left hand side of \eqref{Equation_for_Lower_Bound} can be written as
$$
   \frac{M}{M+1} \log \left( \frac{\mu}{E_B} \right) - \sqrt{\frac{\lambda}{-\mu}} - \sqrt{\frac{\lambda}{\lambda - \mu}} - \alpha(M) \log \left( 1 - \frac{\mu}{\lambda} \right) - \alpha(M),
$$
which is obviously negative for $E_B \leq \mu < 0$. Thus, all solutions of \eqref{Equation_for_Lower_Bound} satisfy $\mu < E_B$.
   \item For fixed $E_B < 0$ and $\lambda > 0$ the left hand side of \eqref{Equation_for_Lower_Bound} is a strictly monotonically decreasing function of $\mu$ on the interval $(-\infty, E_B]$. It tends to $+\infty$ as $\mu \to -\infty$ and it attains a negative value for $\mu = E_B$. Thus, there is a unique solution $\mu < E_B$ of \eqref{Equation_for_Lower_Bound} for fixed $\lambda$ and $E_B$.
   \item The choice of the parameter $\lambda > 0$ is an opportunity for optimization of the lower bound $\mu$.
   \item The fermionic nature of the $N$ identical particles enters the model through the antisymmetric product in the definition of $\HH_N$ and $\HH_{N-1}$ 
and the sign of $\phi^I(E)$ in \eqref{def:Phi}. In fact in the case of $N$ identical bosons, in which an $N$-independent lower bound for $H_N$ cannot be 
established, $\phi^I(E)$ has to be replaced by $-\phi^I(E)$. Therefore, it is important to consider only the negative part of $\phi^I(E)$ when deriving a 
lower bound for $\phi(E)$ in the proof of \Cref{Main_Theorem}. 
 \end{enumerate}
\end{remarks}

Choosing $\lambda = -E_B$ in \Cref{Main_Theorem}, \eqref{Equation_for_Lower_Bound} turns into an equation for $M$ and $\mu/E_B$ only and we obtain the 
following 
statement.

\begin{cor}
 Let $E_B < 0$ and assume that $\alpha(M) < M/(M+1)$. Then,
$$
   H_N \geq \gamma_M \cdot E_B,
$$
where $\gamma_M > 1$ depends on $M$ only and is defined as the unique positive solution of 
\begin{equation*}
  \left( \frac{M}{M+1} - \alpha(M) \right) \log \left( \gamma_M \right) - \frac{1}{\sqrt{\gamma_M}} - \frac{1}{\sqrt{1 + \gamma_M}} - \alpha(M) \log\left( 1 + \frac{1}{\gamma_M} \right) = \alpha(M).
\end{equation*}
\end{cor}

\section{Proofs}

This section is devoted to the proof of \Cref{Main_Theorem}.  We show that $\pi$ times the left hand side of \eqref{Equation_for_Lower_Bound} is a lower bound 
for $\phi(\mu)$ if $\mu<0$. In view of \eqref{Criterion_Lower_Bound}, we then obtain $H_N\geq \mu$ for every solution $\mu$ of \eqref{Equation_for_Lower_Bound}. 
For $r > 0$ let $\chi_r := \chi_{B(0,\sqrt{r})}$ be the characteristic function of the ball $B(0,\sqrt{r})\subset \R^2$. The parameter $\lambda > 0$ is fixed in 
the following and plays the role of an infrared cutoff. For $p^2\leq \lambda$ and $q^2\leq \lambda$ we rewrite the part \eqref{def:Phi_I} of the operator 
$\phi(\mu)$, making use of the pull-through formulas
\begin{equation} \label{Pull_Through_Formula_Entire_Space}
   a_p f(P_f) = f(P_f + p) a_p \qquad \textrm{and} \qquad a_p g(H_f) = g(H_f + p^2) a_p, 
\end{equation}
and the canonical anti-commutation relations, in such a way that
\begin{align}
   \phi(\mu) &= \frac{\pi}{1+\frac{1}{M}} \log \left( \frac{\frac{1}{M+1} P_f^2 + H_f - \mu}{-E_B} \right) + \int\limits_{p^2 \leq \lambda} \!\!\! dp \: \frac{1}{\frac{1}{M} (P_f + p)^2 + H_f + p^2 - \mu} \nonumber\\*
   &\qquad - a(\chi_\lambda) \frac{1}{\frac{1}{M}P_f^2 + H_f - \mu} a^*(\chi_n) - a(\chi_n-\chi_\lambda) \frac{1}{\frac{1}{M}P_f^2 + H_f - \mu} a^*(\chi_\lambda) \nonumber\\*
   &\qquad + \int\limits_{\lambda < p^2, q^2 \leq n} \!\!\!\!\! dp \, dq \: a_p^* \frac{1}{\frac{1}{M} (P_f + p + q)^2 + H_f + p^2 + q^2 - \mu} a_q + o(1) \label{Phi_Rewritten}
\end{align}
as $n \to \infty$. The remainder term converges to zero strongly. Here and in the following $\lambda < p^2, q^2 \leq n$ means that $\lambda < p^2\leq n$ \emph{and}  $\lambda < q^2\leq n$.
The first two terms of \eqref{Phi_Rewritten} are positive for $\mu < E_B$. The last three terms of \eqref{Phi_Rewritten} are estimated, uniformly in $n$ and $N$, in \Cref{Bound_Easy_Terms} and \Cref{Bound_MS}, below. 
	
\begin{lemma}\label{Bound_Easy_Terms}
For $n > \lambda \geq 0$ and $\mu < 0$, \vspace{-0.07cm}
$$
   \norm{(\frac{1}{M}P_f^2 + H_f - \mu)^{-1} a^*(\chi_n-\chi_\lambda)} \leq \sqrt{\frac{\pi}{\lambda - \mu}}.
$$
\end{lemma}

\begin{proof}
 The lemma follows from
\begin{align*}
   \norm{(\frac{1}{M}P_f^2 + H_f - \mu)^{-1} a^*(\chi_n - \chi_\lambda)} 
    &\leq \norm{(H_f - \mu)^{-1} a^*(\chi_n - \chi_\lambda)} \\
    &= \norm{a(\chi_n - \chi_\lambda) (H_f - \mu)^{-2} a^*(\chi_n - \chi_\lambda)}^{1/2},
\end{align*}
and \vspace{-0.1cm}
\begin{align*}
   &a(\chi_n \!-\! \chi_\lambda) \frac{1}{(H_f \!-\! \mu)^2} a^*(\chi_n \!-\! \chi_\lambda) \\
   &= \int\limits_{\lambda < p^2, q^2 \leq n} \!\!\!\!\!\!\!\!\! dp \, dq \: a_p \frac{1}{(H_f - \mu)^2} a_q^*
    = \int\limits_{\lambda < p^2 \leq n} \!\!\!\!\!\! dp \: \frac{1}{(H_f \!+\! p^2 \!-\! \mu)^2} \: - \!\!\!\!\! \int\limits_{\lambda < p^2, q^2 \leq n} \!\!\!\!\!\!\!\!\! dp \, dq \: a_q^* \frac{1}{(H_f \!+\! p^2 \!+\! q^2 \!-\! \mu)^2} a_p \\
   &\leq \int\limits_{p^2 > \lambda} \!\!\! dp \: \frac{1}{(p^2 - \mu)^2} = \frac{\pi}{\lambda - \mu},
\end{align*}
which is true because of the positivity of
\begin{align*}
   &\int\limits_{\lambda<p^2, q^2 \leq n} \!\!\!\!\!\!\!\!\! dp \, dq \,\: a_q^* \: \frac{1}{(H_f + p^2 + q^2 - \mu)^2} \: a_p \\
   & \qquad\qquad = \int\limits_0^\infty \! ds 
\int\limits_0^\infty \! dt \! \int\limits_{\lambda<p^2, q^2 \leq n} \!\!\!\!\!\!\!\!\! dp \, dq \, \: a_q^* \: e^{-(s+t)q^2} e^{-(s+t)(H_f - \mu)} 
e^{-(s+t)p^2} \: a_p.
\end{align*}
\end{proof}

For the proof of \Cref{Bound_MS}, below, we need the following lemma, which is a version of the Schur test.

\begin{lemma}\label{Schur_Test_Variation}
 Let $\Omega \subseteq \R^d$ be a measurable set and let $G: \Omega\times\Omega \to \LL(\FF(L^2(\R^d)))$ be a measurable map. Thus for every $(p,q) \in \Omega\times\Omega$, $G(p,q)$ is a bounded operator on the (antisymmetric) Fock space over $L^2(\R^2)$. Assume that $G(p,q)^* = G(p,q) = G(q,p)$ for all $p,q \in \Omega$. Moreover, let $h: \Omega \to \R_+$ be a positive measurable function. Then,
$$
   \int\limits_{\Omega \times \Omega} \!\!\!\! dp \, dq \: a_p^* \: G(p,q) \: a_q \leq \int\limits_\Omega \!\! dp \: h(p) \: a_p^*  \left( \int\limits_\Omega \!\! dq \: \frac{|G(p,q)|}{h(q)} \right) a_p.
$$
\end{lemma}

\begin{proof}
Let $\psi \in \FF(L^2(\R^d))$. Writing $G(p,q) = \sgn(G(p,q)) \cdot |G(p,q)|$ with the help of the functional calculus, we obtain
\begin{align*}
 &\int\limits_{\Omega \times \Omega} \!\!\!\! dp \, dq \: \sprod{a_p \psi}{G(p,q) \: a_q \psi} 
  \leq \int\limits_{\Omega \times \Omega} \!\!\!\! dp \, dq \: \norm{|G(p,q)|^{1/2} a_p \psi} \cdot \norm{|G(p,q)|^{1/2} \: a_q \psi} \\
  &\leq \left(\:\int\limits_{\Omega \times \Omega} \!\!\!\! dp \, dq \: \frac{h(p)}{h(q)} \norm{|G(p,q)|^{1/2} \: a_p \psi}^2 \right)^{1/2} \left(\:\int\limits_{\Omega \times \Omega} \!\!\!\! dp \, dq \: \frac{h(q)}{h(p)} \norm{|G(p,q)|^{1/2} \: a_q \psi}^2 \right)^{1/2} \\
  &= \int\limits_{\Omega \times \Omega} \!\!\!\! dp \, dq \: h(p) \sprod{\psi}{a_p^* \: \frac{|G(p,q)|}{h(q)} \: a_p \psi}. 
\end{align*}
\end{proof}

\begin{lemma}\label{Bound_MS}
Let $\mu < 0$. Then the operator
$$
   P := \int\limits_{\lambda < p^2, q^2 \leq n} \!\!\!\!\!\!\!\! dp \, dq \: a_p^* \frac{1}{\frac{1}{M} (P_f + p + q)^2 + H_f + p^2 + q^2 - \mu} a_q
$$
admits the estimate
$$
   P \geq - \pi \alpha(M) \left(1 + \log\! \left( 1 + \frac{H_f - \mu}{\lambda} \right) \right).
$$
\end{lemma}

\noindent
\emph{Remark.} Our proof of \Cref{Bound_MS} follows the arguments in \cite{MoserSeiringer}, but in 
contrast to the three-dimensional case, the infrared contributions with $p^2 \leq \lambda$ or $q^2 \leq \lambda$ require a separate treatment.

\begin{proof}
Setting $\shifted{p} := p + \frac{1}{M+2} P_f$ and $\shifted{q} := q + \frac{1}{M+2} P_f$ we can rewrite the denominator in the expression defining $P$ as
$$
   (1+\tfrac{1}{M})(\shifted{p}^2 + \shifted{q}^2) + \tfrac{2}{M} \shifted{p}\cdot\shifted{q} + \tfrac{1}{M+2} P_f^2 + H_f - \mu.
$$
For $\psi \in \bigwedge^{N-1} L^2(\R^2)$, we define $\tilde\psi \in L^2(\R^2; \bigwedge^{N-2} L^2(\R^2))$ by $\tilde{\psi}(p) := a_p \psi$. 
Moreover, we define a unitary operator $T \in \LL(L^2(\R^2; \bigwedge^{N-2} L^2(\R^2)))$ by
$$
   (T \ph)(p;k_1,...,k_{N-2}) := \ph(p + \tfrac{1}{M+2} \sum_{i=1}^{N-2} k_i; k_1,...,k_{N-2}),
$$
where $(T\ph)(p;k_1,...,k_{N-2})$ and $\ph(p;k_1,...,k_{N-2})$ denote values of the functions $(T\ph)(p)$ and $\ph(p) \in \bigwedge^{N-2} 
L^2(\R^2)$, respectively. We obtain
\begin{align*}
 &\sprod{\psi}{P \psi}
   = \int\limits_{\lambda < p^2, q^2 \leq n} \!\!\!\!\!\!\!\! dp \, dq \: \sprod{\tilde\psi(p)}{\frac{1}{(1+\tfrac{1}{M})(\shifted{p}^2 + \shifted{q}^2) + \tfrac{2}{M} \shifted{p}\cdot\shifted{q} + \tfrac{1}{M+2} P_f^2 + H_f - \mu} \tilde\psi(q)} \\
   &= \sprod{(\chi_n - \chi_\lambda) \tilde\psi}{T \sigma T^* (\chi_n - \chi_\lambda) \tilde\psi},
\end{align*}
where $\sigma$ is the operator on $L^2(\R^2; \bigwedge^{N-2} L^2(\R^2))$ with operator-valued integral kernel
$$
   \sigma(p,q) = \frac{1}{(1+\tfrac{1}{M})(p^2 + q^2) + \tfrac{2}{M} p \cdot q + \tfrac{1}{M+2} P_f^2 + H_f - \mu}.
$$
Following \cite{MoserSeiringer} (3.9), we compute the negative part of $\sigma$ explicitly. Its kernel is given by $\sigma^-(p,q) = \frac{1}{2} (\sigma(-p,q) - \sigma(p,q))$. We write $\sigma^-(p,q)$ as
\begin{align*}
   \sigma^-(p,q) &= \frac{1}{2} \left. \frac{1}{(1+\tfrac{1}{M})(p^2 + q^2) - \tfrac{2u}{M} p \cdot q + \tfrac{1}{M+2} P_f^2 + H_f - \mu} \right|_{u=-1}^{u=1} \\
   &= \frac{1}{2} \int\limits_{-1}^1 \! du \: \frac{d}{du} \frac{1}{(1+\tfrac{1}{M})(p^2 + q^2) - \tfrac{2u}{M} p \cdot q + \tfrac{1}{M+2} P_f^2 + H_f - \mu} \\
   &= M p \cdot q \int\limits_{-1}^1 \! du \: \frac{1}{[(M+1)(p^2 + q^2) - 2u p \cdot q + B]^2},
\end{align*}
where $B:= \tfrac{M}{M+2} P_f^2 + MH_f - M\mu$. Then,
\begin{align}
   P &\geq - \!\!\!\! \int\limits_{\lambda < p^2,q^2 \leq n} \!\!\!\!\!\! dp \, dq \: a_p^* \: \sigma^-(\shifted{p}, \shifted{q}) \: a_q \label{P_Sigma} \\
     &= - M \!\!\!\!\!\! \int\limits_{\lambda < p^2,q^2 \leq n} \!\!\!\!\!\! dp \, dq \: a_p^* \left( \:\int\limits_{-1}^1 \! du \: \frac{\shifted{p} \cdot 
\shifted{q}}{[(M+1)(\shifted{p}^2 + \shifted{q}^2) - 2u \shifted{p} \cdot \shifted{q} + B]^2} \right) a_q, \nonumber
\end{align}
and with \Cref{Schur_Test_Variation} and $h(p) = p^2$ we obtain
$$
   P \geq - M \!\!\!\! \int\limits_{\lambda < p^2 \leq n} \!\!\!\!\! dp \; p^2 \:  a_p^* \: f(p, P_f, H_f) \: a_p,
$$
where
$$
   f(p, P_f, H_f) := \int\limits_{\lambda < q^2 \leq n} \!\!\!\!\! dq \, \int\limits_{-1}^1 \! du \: \frac{|\shifted{p} \cdot \shifted{q}|}{q^2 [(M + 
1)(\shifted{p}^2 + \shifted{q}^2) - 2u \shifted{p} \cdot \shifted{q} + B]^2}.
$$
Our goal is now to find a function $g$ with $f(p, Q, E) \leq g(E + p^2)$. It then follows that 
\begin{align}
    P & \geq - M \!\!\!\!\!\! \int\limits_{\lambda < p^2 \leq n} \!\!\!\!\!\!\! dp \; p^2 \: a_p^* \: g(H_f+p^2) \: a_p\nonumber\\
      & \geq - M \int \! dp \; p^2 \: a_p^* a_p \: g(H_f) \geq -M H_f g(H_f).\label{P-larger-g}
\end{align}
To find such a function $g$ we first note that $2 u \shifted{p} \cdot \shifted{q} \leq 0$ on half of the $u$-interval $[-1,1]$ and hence the quotient in the definition of $f$ goes up and becomes independent of $u$ if we drop this term. Second, we use $\shifted{p}^2 + \shifted{q}^2 \geq 2 |\shifted{p} \cdot \shifted{q}|$ and $B \geq 0$ in the denominators. Explicitly,
\begin{align}
 &\int\limits_{-1}^1 \! du \: \frac{|\shifted{p} \cdot \shifted{q}|}{q^2 [(M+1)(\shifted{p}^2 \!+\! \shifted{q}^2) - 2u \shifted{p} \cdot \shifted{q} + B]^2} \nonumber\\*
 &\leq \frac{|\shifted{p} \cdot \shifted{q}|}{q^2 [(M + 1)(\shifted{p}^2 + \shifted{q}^2) + B]^2} + \int\limits_{0}^1 \! du \: \frac{|\shifted{p} \cdot \shifted{q}|}{q^2 [(M + 1)(\shifted{p}^2 + \shifted{q}^2) - 2u |\shifted{p} \cdot \shifted{q}| + B]^2} \nonumber\\
 &\leq \frac{1}{2 q^2 (M+1) [(M+1)(\shifted{p}^2 \!+\! \shifted{q}^2) + B]} + \int\limits_{0}^1 \! du \: \frac{1}{2 q^2 (M+1-u) [(M+1-u)(\shifted{p}^2 \!+\! \shifted{q}^2) + B]}. \label{Preparing_Estimate_Integrand}
\end{align}
One can easily verify that
\begin{equation} \label{Estimate_p2_Without_u}
   (M+1) \shifted{p}^2 + \tfrac{M}{M+2} P_f ^2 \geq \frac{M(M+1)(M+2)}{M^2 + 3M + 1} p^2 \geq M p^2
\end{equation}
and, more generally,
\begin{equation} \label{Estimate_p2_With_u}
   (M+1-u) \shifted{p}^2 + \tfrac{M}{M+2} P_f ^2 \geq \frac{M(M+1-u)(M+2)}{M^2 + 3M + 1 - u} p^2 \geq M \beta(u) p^2,
\end{equation}
where $\beta(u)$ was defined in \eqref{Definition_Beta}.
From \eqref{Preparing_Estimate_Integrand}, \eqref{Estimate_p2_Without_u} and \eqref{Estimate_p2_With_u} we obtain the estimate
\begin{equation} \label{Estimate_f_f_tilde}
   f(p,P_f,H_f) \leq \int \! dq \: \frac{1 - \chi_\lambda(q)}{q^2} \left( \widetilde{f}(\shifted{q},0) + \int\limits_0^1 \! du \: \widetilde{f}(\shifted{q}, u) \right)
\end{equation}
with
$$
   \widetilde{f}(q,u) = \frac{1}{2(M+1-u)^2} \cdot \frac{1}{q^2 + A(u)} \quad \text{and} \quad A(u) = \frac{M[H_f + \beta(u) p^2 - \mu]}{M+1-u}.
$$
In order to estimate \eqref{Estimate_f_f_tilde}, we replace $(1 - \chi_\lambda(q))/q^2$ by the symmetric decreasing function $j_\lambda(q) := (1 - \chi_\lambda(q))/q^2 + \chi_\lambda(q)/\lambda$. We then employ a rearrangement inequality that allows us to replace $\shifted{q} = q + \frac{1}{M+2} P_f$ by $q$ in the argument of $\widetilde{f}$. For an arbitrary $u \in [0,1]$ this reads
\begin{align}
   \int \! dq \: \frac{1 - \chi_\lambda(q)}{q^2} \widetilde{f}(\shifted{q},u)
   &\leq \int \! dq \: j_\lambda(q) \widetilde{f}(q,u) \nonumber \\
   &=\frac{\pi}{2 (M+1-u)^2 A(u)} \left( \frac{A(u)}{\lambda} \log \left( 1 + \frac{\lambda}{A(u)} \right) + \log \left( 1 + \frac{A(u)}{\lambda} \right) \right) \label{Commented_Intermediate_Step} \\
   &\leq \frac{\pi}{2M(M+1-u)\beta(u)} \frac{1}{H_f + p^2} \left(1 + \log \left( 1 + \frac{H_f + p^2 - \mu}{\lambda} \right) \right) \nonumber,
\end{align}
where we used $\log(1+x) \leq x$ if $x \geq 0$ for the first logarithm in \eqref{Commented_Intermediate_Step}, $A(u) \leq H_f + p^2 - \mu$ in the argument of 
the second logarithm and $(M+1-u)A(u) \geq M \beta(u)(H_f + p^2)$ in the overall prefactor $1/A(u)$. Combining \eqref{Estimate_f_f_tilde} and 
\eqref{Commented_Intermediate_Step}, we arrive at
$$
   f(p, P_f, H_f) \leq \frac{\pi \alpha(M)}{M} \frac{1}{H_f + p^2} \left( 1 + \log \left( \frac{H_f + p^2 - \mu}{\lambda} \right) \right),
$$
which is of the form $g(H_f + p^2)$ as desired. In view of \eqref{P-larger-g} the lemma is proven.
\end{proof}

\begin{proof}[Proof of \Cref{Main_Theorem}]
We combine \eqref{Phi_Rewritten}, \Cref{Bound_Easy_Terms}, \Cref{Bound_MS} and $\norm{a(\chi_\lambda)} = \norm{a^*(\chi_\lambda)} = \sqrt{\pi \lambda}$. In the limit $n \to \infty$ we find
\begin{align*}
   \phi(\mu)
   &\geq \frac{\pi}{1+\frac{1}{M}} \log \left( \frac{\frac{1}{M+1} P_f^2 + H_f - \mu}{-E_B} \right) - \pi \sqrt{\frac{\lambda}{-\mu}} - \pi \sqrt{\frac{\lambda}{\lambda - \mu}} \\*
   &\qquad - \pi \alpha(M) \left(1 + \log\! \left( 1 + \frac{H_f - \mu}{\lambda} \right) \right) \\
   &\geq \pi \left( \frac{M}{M+1} - \alpha(M) \right) \log \left( \frac{\mu}{E_B} \right) - \pi \sqrt{\frac{\lambda}{-\mu}} - \pi \sqrt{\frac{\lambda}{\lambda - \mu}} \\*
   &\qquad - \pi \alpha(M) \log\left( -E_B \left( \frac{1}{\lambda} + \frac{1}{-\mu} \right) \right) - \pi \alpha(M).
\end{align*}
By \eqref{Criterion_Lower_Bound}, this completes the proof of \Cref{Main_Theorem}.
\end{proof}

It would be very interesting to know whether the conclusion of \Cref{Main_Theorem} still holds for $M < 1.225$. One could address this question with the help 
of some numerics as follows. Using the pull-through formula and \Cref{Schur_Test_Variation} with $h(p) = p^2$ one obtains from \eqref{P_Sigma}
\begin{align*}
 &\sqrt{\tfrac{H_f-\mu}{\log ( 1 + \frac{H_f-\mu}{\lambda})}} P \sqrt{\tfrac{H_f-\mu}{\log ( 1 + \frac{H_f-\mu}{\lambda})}} \\
 &\qquad \geq - \!\!\! \int\limits_{\lambda < p^2 \leq n} \!\!\! dp \: p^2 a_p^*  \left( \: \int\limits_{\lambda < q^2 \leq n} \!\!\!\!\!\! dq \: \frac{1}{q^2} 
\: \sqrt{\tfrac{H_f+p^2-\mu}{\log ( 1 + \frac{H_f+p^2-\mu}{\lambda})}} |\sigma^-(\shifted{p}, \shifted{q})| \sqrt{\tfrac{H_f+q^2-\mu}{\log ( 
1 + \frac{H_f+q^2-\mu}{\lambda})}} \right) \: a_p \\
 &\qquad \geq - C H_f,
\end{align*}
with
\begin{align*}
 C&:= \sup_{\substack{p,Q \in \R^2 \\  \tau > 0}} \sqrt{\tfrac{\tau+p^2-\mu}{\log ( 1 + \frac{\tau+p^2-\mu}{\lambda})}} \times \\
  &\qquad\qquad\quad \int\limits_{\lambda < q^2} \!\! dq \: \frac{1}{q^2} \sqrt{\tfrac{\tau+q^2-\mu}{\log (1 + \frac{\tau+q^2-\mu}{\lambda})}} 
\frac{\frac{2}{M} |\shiftedB{p} \cdot \shiftedB{q}|}{((1+\frac{1}{M})(\shiftedB{p}^2 + \shiftedB{q}^2) + \frac{1}{M+2}Q^2 + \tau - \mu)^2 - \frac{4}{M^2} 
(\shiftedB{p} \cdot \shiftedB{q})^2}
\end{align*}
where $\shiftedB{p}:= p + \frac{1}{M+2} Q$ and $\shiftedB{q}:= q + \frac{1}{M+2} Q$. This yields $P \geq - C \cdot \log ( 1 + \frac{H_f-\mu}{\lambda})$. One 
could now attempt to evaluate the constant $C$ numerically and compare it with the prefactor $\pi/(1+\frac{1}{M})$ in the definition of $\phi^0(E)$ given in 
\eqref{def:Phi_0}. A corresponding numerical analysis was done successfully in the three-dimensional case \cite{MoserSeiringer}.

\appendix
\section{Appendix}\label{sec:appendix}
In this appendix we briefly explain the connection between $\phi(E)$ defined in the introduction, the Hamiltonian $H_N$ that occurs in \eqref{Criterion_Lower_Bound}, and \eqref{Formal_Expression_1}, see also Section~5.1 of \cite{Thesis}. 

Let $H_0 := M^{-1}P_f^2 +H_f$, and for $E<0$ let $R_E:=V(H_0-E)^{-1}\in \LL(\HH_{N},\HH_{N-1})$, where $V: D(H_0) \cap \HH_N \to \HH_{N-1}$ is defined by
$$
   V\psi := \lim_{n \to \infty}  \int\limits_{k^2 \leq n} \!\! dk \: a_k \psi.
$$
The existence of this limit is easily established with the help of the pull-through formula $a_k(H_0-E)^{-1} = (H_0+k^2-E)^{-1}a_k$  \cite{Thesis,DR}. The domain $D(H_{N})$ of $H_N$ can be characterised as follows: a vector $\psi\in \HH_N$ belongs to $D(H_{N})$ if and only if there is a vector $w_\psi \in D(\phi) \subseteq \HH_{N-1}$ such that for some (and hence all) $E < 0$
\begin{equation} \label{H_Condition_1}
   \psi - R_E^* w_\psi \in D(H_0),
\end{equation}
and
\begin{equation} \label{H_Condition_2}
   V(\psi - R_E^* w_\psi) = \phi(E) w_\psi.
\end{equation}
For $\psi\in D(H_N)$ the action of $H_N$ is given by
\begin{equation} \label{Action_of_H}
   (H_N - E)\psi = (H_0 - E) (\psi - R_E^* w_\psi).
\end{equation}
By \eqref{H_Condition_2}, \eqref{Action_of_H}, and the definition of $R_E$,
\begin{align*}
   \sprod{\psi}{(H_N - E) \psi} &= \sprod{\psi - R_E^* w_\psi}{(H_0 - E)(\psi - R_E^* w_\psi)} + \sprod{w_\psi}{\phi(E) w_\psi}
\end{align*}
which proves Condition \eqref{Criterion_Lower_Bound}. 

The Hamiltonian $H_N$ as described above is a self-adjoint operator \cite{Five_Italians, Thesis} and it represents the formal expression 
\eqref{Formal_Expression_1} in the center-of-mass frame, or, which is the same, in the sector of total momentum zero. To explain this let us rewrite 
\eqref{Formal_Expression_1} in terms of center-of-mass and relative coordinates, $R = (My + \sum_{i=1}^N x_i)/(M+N)$ and $r_i = x_i - y$, respectively. One 
obtains the sum of the kinetic energy of the center-of-mass motion, $-(M+N)^{-1} \Delta_R$, and 
\begin{equation} \label{Formal_Expression_2}
   \frac{1}{M} \left( \sum_{i=1}^N i \nabla_{r_i} \right)^2 - \sum_{i=1}^N \Delta_{r_i} - g \sum_{i=1}^N \delta(r_i).
\end{equation}
Here we recognise in the first two terms the free Hamiltonian $H_0$. We expect that $H_N$ agrees with $H_0$ away from the support of the $\delta$-potentials. Indeed, for $\psi\in D(H_0)\cap\ker (V)$ we may choose $w_\psi=0$. It follows that $\psi\in D(H_N)$ and that $H_N\psi=H_0\psi$. Thus $H_N$ is an extension of $H_0$ restricted to $D(H_0)\cap\ker (V)$. Now, for a smooth function $\psi\in \HH_N$ the condition $\psi\in \ker(V)$ is equivalent to $\psi(x_1,\ldots,x_N)=0$ whenever $x_k=0$ for some $k$.
In the literature, an extension of $H_0$, characterized by a condition of the form \eqref{H_Condition_2}, is known as Skornyakov-Ter-Martirosyan (STM) 
extension. In the analogous situation in three dimensions with suitable values of the system parameters, a variety of different STM extensions is known to 
exist \cite{Non_STM_Extension}.

\bigskip\noindent
\textbf{Acknowledgements:} Ulrich Linden thanks Robert Seiringer and Thomas Moser for encouraging discussions and the hospitality at the IST Austria. His work was
supported by the \emph{Deutsche Forschungsgemeinschaft (DFG)} through the Research Training Group 1838: \emph{Spectral Theory and Dynamics of Quantum Systems}. 
 
\bibliographystyle{plain}

\end{document}